\documentclass[10pt]{article}

\usepackage{amsthm}
\usepackage{amssymb,amsmath}
\usepackage{amsfonts}
\usepackage[T1]{fontenc}
\usepackage[ruled,linesnumbered,noline,noend]{algorithm2e}
\usepackage{tikz}
\usetikzlibrary{shapes,arrows}

\newtheorem{thm}{Theorem}%[section]
\newcommand{\false}{\ensuremath{\mathit{false}}}
\newcommand{\true}{\ensuremath{\mathit{true}}}
\newcommand{\ite}{\ensuremath{\mathbf{ite}}}

\newcommand{\var}[1]{{\tt #1}}
\newcommand{\sym}[1]{\ensuremath{\underline{#1}}}

\newcommand{\Pex}{\textsc{Pex}\xspace}
\newcommand{\Sage}{\textsc{Sage}\xspace}
\newcommand{\Cute}{\textsc{Cute}\xspace}
\newcommand{\Klee}{\textsc{Klee}\xspace}
\newcommand{\Exe}{\textsc{Exe}\xspace}
\newcommand{\APC}{\textsc{Apc}\xspace}

\newcommand{\Z}{\textsc{Z3}\xspace}

\DontPrintSemicolon
\newcommand{\aarg}[2]{#1~{\it //~#2}}
\newcommand{\aargm}[2]{\\ \qquad \aarg{#1}{#2}}
\newcommand{\aset}{\ensuremath{\longleftarrow}}

\begin{document}
 
%{{{ Title + Abstract 

\title{Abstracting Path Conditions}

\author{
   Jan Strej\v{c}ek ~~~~~~~ Marek Trt\'{\i}k \\
   {\small Faculty of Informatics} \\
   {\small Masaryk University}\\
   {\small Brno, Czech Republic}\\
	 {\small \texttt{\{strejcek,trtik\}@fi.muni.cz}}
}

\date{December 18, 2011}

\maketitle

\begin{abstract}
  We present a symbolic execution based algorithm that for a given program
  and a given program location produces a nontrivial necessary condition on
  input values to drive the program execution to the given location. We
  propose a usage of the produced condition in
  contemporary bug finding and test generation tools based on symbolic
  execution. Experimental results indicate that the presented technique can
  significantly improve performance of the tools.
\end{abstract}

%\category{D.2.4}{Software Engineering}{Software/Program Verification}
%\category{D.2.5}{Software Engineering}{Testing and Debugging}
%\category{F.3.1}{Logics and Meanings of Programs}{Specifying and Verifying
%  and Reasoning about Programs}
%
%\terms{Reliability, Verification, Code coverage, Tests generation, Bug-finding}
%
%\keywords{Symbolic execution, Path condition, Backbone paths, Loop summaries}

%}}}
%{{{ Introduction

\section{Introduction} 

Symbolic execution~\cite{BEL75,Kin76,How77} is enjoying a renaissance during
the last decade. The basic idea of the technique 
% originally formulated in seventies
is to replace input data of a program by symbols representing arbitrary
data. Executed instructions then manipulate expressions over the symbols
rather than exact values. A symbolic execution produces, for each path in a
program flowgraph starting in the initial location, a formula called
\emph{path condition}, i.e.~the necessary and sufficient condition on input
data to drive the execution along the path. Symbolic execution is utilized
by many successful algorithms and tools for test generation and bug finding,
for example \Exe~\cite{Cadar08}, \Cute~\cite{SMA05}, \Klee~\cite{CDE08}, \Sage~\cite{GLM08:fuzzing}, or \Pex~\cite{TdH08}. These
tools can relatively quickly find tests that cover vast majority of a given
code. However, they usually fail to cover the code completely in a
reasonable time.
% A typical pattern of a program location that is hard to cover is the last
% line of the following code.
% \begin{verbatim}
% a = 3;
% for (i = 0; i < n; ++i) {
%   ++a;
% if (a > 100)
%   assert(false);
% \end{verbatim}
% The problem is that validity of the condition guarding the last line depends
% on the number of iterations of the loop. 
In this paper we suggest a method that helps the tools to cover a chosen
location and hence to further improve their performance.

The core of our method and the main contribution of the paper is an
algorithm that, for a given program and a given program location, produces a
nontrivial necessary condition on input values to drive the program
execution to the given location. An intuitive explanation of the algorithm
is illustrated on the following simple C++ program, where we want to compute a
necessary condition to reach the assertion on the last line.

\begin{verbatim}
void foo(int* A, int n) {
   int k = 3;
   for (int i = 0; i < n; ++i) {
      if (A[i] == 1)
         ++k;
   }
   if (k > 12)
      assert(false);
}
\end{verbatim}
% The program computes the number of ones in the first \texttt{n} elements of
% array \texttt{A}. If three plus the number greater than 12, then the
% assertion is visited. 
It is easy to check (for human) that the assertion is reached when there is more then twelve numbers $1$ in array \texttt{A}. Figure~\ref{fig:ex1}~(a) depicts a flowgraph of \texttt{foo}. Note that nodes and edges that are not on any path to the target location $h$ have been removed.
% We mark locations (nodes) as $1,2,\ldots$.

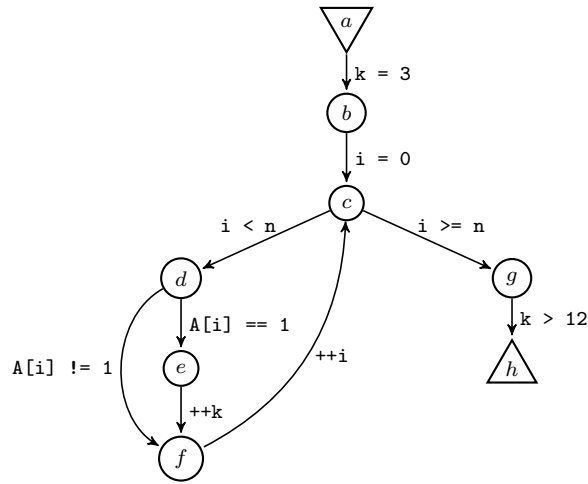
\begin{figure}[ht]
  \centering
  % Define block styles
  \tikzstyle{start} = [regular polygon,regular polygon sides=3,
    regular polygon rotate=180,thick,draw,inner sep=1.7pt]
  \tikzstyle{target} = [regular polygon,regular polygon sides=3,
    regular polygon rotate=0,thick,draw,inner sep=1pt]
  \tikzstyle{loc} = [circle,thick,draw]
  \tikzstyle{pre} = [<-,shorten <=1pt,>=stealth',semithick]
  \tikzstyle{post} = [->,shorten <=1pt,>=stealth',semithick]
  \footnotesize
  \begin{tikzpicture}[node distance=1.2cm]
    \node [start] (a) {$a$};
    \node [loc] (b) [below of=a] {$b$}
      edge [pre] node [label=right:\texttt{k = 3}] {} (a);
    \node [loc] (c) [below of=b] {$c$}
      edge [pre] node [label=right:\texttt{i = 0}] {} (b);
    \node [loc] (d) [left of=c,xshift=-10mm,yshift=-10mm] {$d$}
      edge [pre] node [label=above:\texttt{i < n~~~~}] {} (c);
    \node [loc] (e) [below of=d] {$e$}
      edge [pre] node [label=right:\texttt{A[i] == 1}] {} (d);
    \node [loc] (f) [below of=e] {$f$}
      edge [pre,bend left=60] node [label=left:\texttt{A[i] != 1}] {} (d)
      edge [pre] node [label=right:\texttt{++k}] {} (e)
      edge [post,bend right] node [label=right:\texttt{++i}] {} (c);
    \node [loc] (g) [right of=c,xshift=10mm,yshift=-10mm] {$g$}
      edge [pre] node [label=above:\texttt{~~~~i >= n}] {} (c);
    \node [target] (h) [below of=g] {$h$}
      edge [pre] node [label=right:\texttt{k > 12}] {} (g);
  \end{tikzpicture}
  \caption{Flowgraph of the running example.}
  \label{fig:ex1}
\end{figure}

As the first step of our algorithm, we find all nontrivial maximal strongly
connected components in the flowgraph. For each entry node $x$ of each
component (i.e~there is an edge leading to $x$ from a vertex outside the component), we compute a \emph{summary} of the overall effect of
iterating within the component, since the first visit of $x$ till the last
visit of $x$. The summary is described by an \emph{iterated symbolic state} and \emph{looping condition}. An iterated symbolic state is a function that assigns to each program variable its value given by an
expression over \emph{symbols} and \emph{path counters}. Symbols represent
initial values of variables (for each variable $\var{v}$ the symbol is
denoted by $\sym{v}$). Path counters $\kappa_1,\kappa_2,\ldots$ correspond
to different acyclic paths leading from $x$ to $x$ within the
component. Each path counter represents the number of iterations of the corresponding path. A looping condition is a nontrivial formula implied by any of path conditions resulting from any symbolic execution of the component.

In our example, there is only one nontrivial maximal strongly connected component
$\{c,d,e,f\}$ with one entry node $c$. There are two acyclic paths through
the component: $\pi_1=\mathit{cdefc}$ and $\pi_2=\mathit{cdfc}$. We assign
path counters $\kappa_1,\kappa_2$ to $\pi_1,\pi_2$ respectively. The overall
effect of the component with respect to the entry node $c$ can be described
by the iterated symbolic state $\theta^{\vec{\kappa}}$ with only two
interesting values (as the other variables are not changed in the
component):
% $\theta^{\vec{\kappa}}$ assigns them the corresponding symbols):
\[
\begin{array}{rcl}
\theta^{\vec{\kappa}}(\var{i})&=&\kappa_1+\kappa_2+\sym{i}\\
\theta^{\vec{\kappa}}(\var{k})&=&\kappa_1+\sym{k}\\
\end{array}
\]
In other words, by $\kappa_1$ iterations of $\pi_1$ and $\kappa_2$
iterations of $\pi_2$ executed in an arbitrary order, the values of
$\var{i}$ and $\var{a}$ are increased by $\kappa_1+\kappa_2$ and $\kappa_1$,
respectively.

Further, for every component and its entry node $x$ we compute a
looping condition $\varphi^{\vec{\kappa}}$. Given path counters
$\kappa_1,\kappa_2,\ldots$, formula $\varphi^{\vec{\kappa}}$ describes a
necessary condition to keep looping in the component for $\sum_i\kappa_i$
iterations such that, for each $i$, exactly $\kappa_i$ iterations use
path $\pi_i$. More precisely, a looping condition is a conjunction of
subformulae $\varphi_i$ corresponding to the acyclic paths $\pi_i$. Each
subformula $\varphi_i$ says that, for each of the $\kappa_i$ iterations
along the path $\pi_i$, all tests on the path must be satisfied for some
possible values of variables, i.e.~for some values given by the iterated
symbolic state and some admissible values of path counters.

In the example, the looping condition for the component $\{c,d,e,f\}$ with
the entry node $c$ has the form
$\varphi^{\vec{\kappa}}=\varphi_1\wedge\varphi_2$. We focus on the
construction of $\varphi_1$ which corresponds to path $\pi_1=\mathit{cdefc}$ with two
tests: \texttt{i < n} and \texttt{A[i] == 1}. The iterated symbolic state
says that values of $\var{i}$, $\var{n}$, and $\texttt{A[i]}$ in
$(\tau_1+1)$-st iteration of $\pi_1$ and after $\tau_2$ iterations of
$\pi_2$ are $\tau_1+\tau_2+\sym{i}$, $\sym{n}$, and
$\sym{A}(\tau_1+\tau_2+\sym{i})$ respectively. Hence, if we want to make
$\kappa_1$ iterations of $\pi_1$ and $\kappa_2$ iterations of $\pi_2$, the
formula $\varphi_1$ says that for each $\tau_1$ satisfying
$0\le\tau_1<\kappa_1$ there has to be some $\tau_2$ satisfying
$0\le\tau_2\le\kappa_2$ such that $\tau_1+\tau_2+\sym{i}<\sym{n}$ and
$\sym{A}(\tau_1+\tau_2+\sym{i})=1$. The complete looping condition for our
example is as follows:
\[
\setlength{\arraycolsep}{2pt}
\begin{array}{rcl}
  \varphi^{\vec{\kappa}}&\equiv&\varphi_1\wedge\varphi_2\\[1ex]
  \varphi_1&\equiv&\forall\tau_1\big(0\le\tau_1<\kappa_1\rightarrow\exists\tau_2(0\le\tau_2\le\kappa_2~\wedge\\
  &&~~~~~~~~\wedge~\tau_1+\tau_2+\sym{i}<\sym{n}~~\wedge~~\sym{A}(\tau_1+\tau_2+\sym{i})=1)\big)\\[1ex]
  \varphi_2&\equiv&\forall\tau_2\big(0\le\tau_2<\kappa_2\rightarrow\exists\tau_1(0\le\tau_1\le\kappa_1~\wedge\\
  &&~~~~~~~~\wedge~\tau_1+\tau_2+\sym{i}<\sym{n}~~\wedge~~\sym{A}(\tau_1+\tau_2+\sym{i})\neq1)\big)
\end{array}
\]

% \begin{center}
%   \begin{tabbing}
%     $\varphi^{\vec{\kappa}}~=~$\=$\forall\tau_1\big($\=$0\le\tau_1<\kappa_1 \implies$\\
%     \>\>$\exists\tau_2($\=$0\le\tau_2\le\kappa_2~\wedge\,$\\
%     \>\>\>$\tau_1+\tau_2+\sym{i}<\sym{n}~~\wedge~~\sym{A}(\tau_1+\tau_2+i)=1)\big)~\wedge$\\
%     \>$\forall\tau_2\big($\=$0\le\tau_2<\kappa_2 \implies$\\
%     \>\>$\exists\tau_1($\=$0\le\tau_1\le\kappa_1~\wedge\,$\\
%     \>\>\>$\tau_1+\tau_2+\sym{i}<\sym{n}~~\wedge~~\sym{A}(\tau_1+\tau_2+i)\neq1)\big)$
%   \end{tabbing}
% \end{center}

The resulting summary of the component is a pair $(\theta^{\vec{\kappa}}, \varphi^{\vec{\kappa}})$. We attach the summary at the entry node $c$ and we can proceed to analysis of 
%in the following. We modify the flowgraph in such a way, that each strongly
%connected component has only one entry node (we make as many copies of each
%component as needed). To each entry node we associate the corresponding
%iterated symbolic state and looping condition. % \textbf{Rozepsat to jako ze
%% to nahradime dvouma hranama a na prvni dame looping condition (jako by tam
%% proste byl dalsi test) a na druhou dame symbolic memory (jako by to bylo
%% prirazeni)?}
%Then we break each component by removing all edges leading from nodes in the
%component to the entry node. Breaking a strongly connected component can
%produce new strongly connected components corresponding to nested cycles. We
%repeat this procedure until the flowgraph is acyclic. Then we apply (a
%slight modification of) symbolic execution to get an \emph{abstract path
%  condition} for each path from the initial node to the target node. The
%path condition is called abstract as it represents not only the path it has
%been computed for, but also all paths that differ from this path by looping
%in the removed components. The output of our technique is a disjunction of
%all abstract path conditions.
%In our running example, the described procedure leaves only one path from
%$a$ to $h$ in the flowgraph, namely
the path $\mathit{abcgh}$ in the flowgraph. We symbolically execute the path as usual. Only at loop entry $c$ we add the saved summary into the current symbolic state and current path condition. The abstract
path condition (and thus also the final result of our technique) is the
following formula $apc$, where $\varphi^{\vec{\kappa}}[\sym{i}/0,\sym{a}/3]$
is the looping condition computed above with $\sym{i}$ replaced by $0$ and
$\sym{a}$ replaced by $3$.
\[
\setlength{\arraycolsep}{0pt}
\begin{array}{rl}
  apc\equiv\exists\kappa_1,\kappa_2(
  &\kappa_1,\kappa_2>0~\,\wedge\,~\varphi^{\vec{\kappa}}[\sym{i}/0,\sym{a}/3]~\,\wedge\\
  &\wedge\,~\kappa_1+\kappa_2\ge n~\wedge~\kappa_1+3>12)
\end{array}
\]

To sum up, our technique produces a formula $apc$ that has to be
satisfied by all inputs driving the execution to the given location. In
general, the formula is not a sufficient condition on inputs to reach the
target location. This has basically two reasons. 
\begin{itemize}
\item It is not always possible to express the overall effect of a strongly
  connected component to a variable in a declarative way. In such a case,
  the variable is assigned the special value $\star$ with the meaning
  ``unknown''.  If we symbolically execute a test containing a variable with
  the value $\star$, we do not add this test to our abstract path
  condition. Similarly, the tests containing $\star$ are not added to
  looping conditions.
\item The looping condition is constructed as a necessary but not a sufficient
  condition. More precisely, it checks whether tests in each iteration are
  satisfied for the iterated symbolic state with some admissible values of
  path counters, but the consistency of these admissible values over all
  iterations is not checked.
  % It is possible to construct the looping condition to be necessary and
  % sufficient, but finding a satisfactory model
  % of such a formula would be much harder. The chosen form of looping
  % condition seems to be a reasonable trade-off between precision and
  % simplicity.
\end{itemize}

In the following sections, we explain our algorithm in more detail. After
providing some preliminaries (Section~\ref{sec:prelim}), we present the
basic version of the algorithm for flowgraphs with integer arithmetic and
read-only multi-dimensional arrays and without function calls
(Section~\ref{sec:alg}). Then we indicate necessary changes to the algorithm
to work with programs that can modify arrays (Section~\ref{sec:alg2}). In
the same way as arrays, the algorithm can also handle flowcharts
manipulating content of lists (we currently do not support programs changing
shape of lists). %Hence, the extension makes the algorithm applicable
%on an interesting class of real code. For example, the algorithm can be
%easily adopted for C++ code manipulating arrays and lists by some functions
%of Standard Template Library like \texttt{copy}, \texttt{find},
%\texttt{transform}, \texttt{for\_each}, \texttt{count} etc.
To demonstrate
efficiency of our approach, we provide experimental results of a
prototype implementation of our algorithm on several small examples
(Section~\ref{sec:experiments}). The results show that in some cases, the
application of our algorithm can discover a bug in a code much faster than
selected bug finding tools do. Therefore we suggest possible utilization our
algorithm in contemporary bug finding and test generation tools
(Section~\ref{sec:apps}). Finally, we discuss some related work
(Section~\ref{sec:related}) and conclude the paper
(Section~\ref{sec:conclusion}).

%}}}
%{{{ Preliminaries

\section{Preliminaries}\label{sec:prelim}

This section defines some terms heavily used in the rest of the paper, in
particular terms related to \emph{program} and \emph{symbolic state}.

\subsection{Program}
The algorithm works with programs in the form of flowcharts. A target
location is a distinguished node of the flowchart and it has no
successor. Moreover, we assume that the flowchart contains only nodes from
which the targer node is reachable. Formally, a \emph{program} is a tuple
$P=(V_P,E_P,l_s,l_t,\iota_P)$ such that $(V_P,E_P)$ is a finite connected
oriented graph, nodes $V_P$ represent program locations, edges $E_P\subseteq
V_P\times V_P$ represent control flow between them, $l_s,l_t\in V_P$, $l_s
\ne l_t$ are \emph{start} and \emph{target nodes} respectively. A node is
\emph{branching} if its out-degree is $2$. All other nodes, except $l_t$,
have out-degree $1$. In-degree of $l_s$ and out-degree of $l_t$ are both
$0$. Function $\iota_P:E_P\rightarrow\mathcal{I}$ assigns to each edge $e$
an \emph{instruction} $\iota(e)$. We use two kinds of instruction: an
assignment instruction $\var{v}\aset e$ for some scalar variable $\var{v}$
and some expression $e$, and an assumption $\texttt{assume}(\gamma)$ for
some quantifier-free formula $\gamma$ over program variables. Out-edges of
any branching node are labelled with instructions $\texttt{assume}(\gamma)$
and $\texttt{assume}(\neg\gamma)$ for some $\gamma$. Further, we assume that
all instructions in $\mathcal{I}$ use only linear integer arithmetic. %We use
%names $\var{a},\var{b},\ldots$ for scalar variables of type \texttt{Int} and
%$\var{A},\var{B},\ldots$ for array variables of type
%\texttt{Int}$^k\rightarrow$\texttt{Int}.
By $\mathcal{V}_\var{a}$ and
$\mathcal{V}_\var{A}$ we denote the sets of all \emph{scalar variables} and
\emph{array variables} occurring in $P$, respectively. And $\mathcal{V} = \mathcal{V}_\var{a} \cup \mathcal{V}_\var{A}$ is a set of all variables.
% We can supply a precondition $\varphi$ and a postcondition $\psi$ for $P$
% by introducing new vertices $l_s, l_t$ and connecting them to old ones by
% two edges. And the labelling of the only out-edge of $l_s$ and of the only
% in-edge to $l_t$ are $\texttt{assume}(\varphi)$ and
% $\texttt{assert}(\psi)$ instructions respectively.
When program $P$ is clearly determined by a context, we omit the subscript
$P$ in $V_P,E_P,\iota_P$.

% \paragraph{Assertions and assumptions in a program} Suppose that we execute
% symbolically a program. Let $\varphi$ be a path condition. Then
% $\texttt{assert}(\gamma)$ forces validity check of formula $\varphi
% \rightarrow \gamma$. The execution may continue only if the check
% succeeded. Note that the path condition is not updated. On the other hand
% $\texttt{assume}(\gamma)$ updates the path condition such that $\varphi
% \aset \varphi \wedge \gamma$ and the execution may continue, if updated
% $\varphi$ is satisfiable.

A \emph{path} in a program is a finite sequence $\pi = v_1 v_2 \cdots v_k$
of program nodes such that $(v_i, v_{i+1})\in E$ for all $1\le i<k$. Paths
are always denoted by greek letters. A path leading from $l_s$ to $l_t$ is
called \emph{complete} path.

Instead of strongly connected components, our algorithm works with
\emph{loops}. In contrast to components, loops can be nested. Let $\pi$ be
an acyclic path from the initial node $l_s$ and let $\alpha$ be a prefix of
$\pi$ leading to a node $v$. The node $v$ on $\pi$ is an \emph{entry node}
of a loop if there exists a path $v\beta v$ such that none of the nodes on
$\beta$ appears in $\alpha$. The entry node $v$ on $\pi$ enters the loop
$C$ that is the smallest set containing $v$ and all nodes in
$\beta$ for each path $v\beta v$ such that none of the nodes on $\beta$
appears in $\alpha$. For example, program in Figure~\ref{fig:loops} contains
two acyclic complete paths: $\pi_1=\mathit{l_sbdl_t}$ and
$\pi_2=\mathit{l_sabdl_t}$. While $\pi_1$ contains only one entry node $b$
associated with loop $\{a,b,c,d\}$, $\pi_2$ contains entry node $a$ with
loop $\{a,b,c,d\}$ and entry node $b$ with loop $\{b,c\}$.
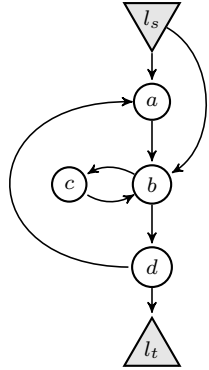
\begin{figure}[tb]
  \centering
  % Define block styles
  \tikzstyle{start} = 
    [regular polygon,regular polygon sides=3,thick,draw,inner sep=1pt]
  \tikzstyle{loc} = [circle,thick,draw]
  \tikzstyle{pre}=[<-,shorten <=1pt,>=stealth',semithick]
  \tikzstyle{post}=[->,shorten <=1pt,>=stealth',semithick]
  \footnotesize
  \begin{tikzpicture}[node distance=1.1cm]
    \node [start,fill=black!10,regular polygon rotate=180] (ls) {$l_s$};
    \node [loc] (a) [below of=ls] {$a$}
      edge [pre] (ls);
    \node [loc] (b) [below of=a] {$b$}
      edge [pre] (a)
      edge [pre,bend right=60] (ls);
    \node [loc] (c) [left of=b] {$c$}
      edge [pre,bend left] (b)
      edge [post,bend right] (b);
    \node (x) [left of=c] {}; 
    \node [loc] (d) [below of=b] {$d$}
      edge [pre] (b)
      edge [post,bend left=90,looseness=2.5] (a);
    \node [start,fill=black!10] (lt) [below of=d] {$l_t$}
      edge [pre] (d);
  \end{tikzpicture}
  \caption{Example of nested loops.}
  \label{fig:loops}
\end{figure}
A node $u$ is an exit node of $C_v$ if there exists $w \in
C_v$ such that $(w,u)\in E$. 

For a loop $C$ with an entry node $v$, a \emph{program induced by the loop},
denoted as $P(C,v)$, is the subgraph of the original program induced by $C$
where $v$ is marked as the start node, a fresh node $v'$ is added and marked
as the target node, and every edge $(u,v)\in E$ leading to $v$ is replaced
by an edge $(u,v')$.

Let $\pi$ be a complete path. We define a \emph{backbone} of $\pi$ as the
result of the following procedure: If $\pi$ is acyclic, then the backbone is
directly $\pi$. Otherwise, $\pi$ can be written as $\alpha v\beta v\gamma$
where $v$ is the first repreating node in $\pi$ and $\gamma$ does not
contain $v$. In this case, we set $\pi$ to $\alpha v\gamma$ and repreat the
procedure. By $B_P$ we denote the set of all backbones of all complete
paths. Note that backbones are exactly all acyclic complete
paths. Alternatively, backbone of a complete path $\pi$ can be defined as
the path from $l_s$ where the successor of each node $u$ is the same as the
successor of the last occurence of $u$ in $\pi$.

\subsection{Symbolic State}
The set $\mathcal{S}$ of \emph{symbolic expressions} contains all
expressions build with integers, standard integer operations and functions,
and
\begin{itemize}
\item a \emph{constant symbol} $\sym{a}$ for each scalar variable
  $\var{a}\in\mathcal{V}_\var{a}$,
\item a \emph{function symbol} $\sym{A}$ for each array variable
  $\var{A}\in\mathcal{V}_\var{A}$, where $\sym{A}$ has the same arity as $\var{A}$,
\item a countable set $\{\kappa_1,\tau_1,\kappa_2,\tau_2,\ldots\}$ of
  variables called \emph{path counters},
\item a special construct $\ite(\varphi,e_1,e_2)$ where $e_1,e_2$ are
  expressions and $\varphi$ is a first order formula over the same signature
  extended with standard relation symbols, and
\item a special constant symbol $\star$ called \emph{unknown}.
\end{itemize}
% By $\sym{V}$ we denote the set of all constant and function symbols. 
The value of $\ite(\varphi,e_1,e_2)$ is the same as $e_1$ if $\varphi$ holds
and the same as $e_2$ otherwise. The domain of integers is extended with a
new special value $\bot$. All expressions containing $\star$ are interpreted
to $\bot$ while the other expressions are never interpretted to $\bot$. In
the following we identify every expression containing $\star$ with $\star$.

Let $f,e_1,e_2,\ldots,e_n$ be symbolic expressions and $x_1,x_2,\ldots,x_n$
be some path counters or constant symbols corresponding to scalar
variables. Then $f[x_1/e_1,$ $x_2/e_2,\ldots,x_n/e_n]$ is an symbolic
expression $f$ where all occurences of $x_i$ are replaced by $e_i$,
simultaneously for all $i$. To shorten the notation, we also write
$f[\vec{x}/\vec{e}]$ when the meaning is clerly given by a context. we also
use the notation $\varphi[\vec{x}/\vec{e}]$ with the analogous meaning.

A \emph{symbolic state} is a function
$\theta:\mathcal{V}\rightarrow\mathcal{S}$ assigning to each variable \var{a} a symbolic expression $\theta(a)$ . We define \emph{initial
  symbolic state} $\theta_I$ and \emph{unknown symbolic state}
$\theta_\star$ as
\[\theta_I(\var{a})=\sym{a},~\theta_I(\var{A})=\lambda\vec{\chi}.\sym{A}(\vec{\chi}) \textrm{~~and~~}\theta_\star(\var{a})=\star,~\theta_\star(\var{A})=\lambda\vec{\chi}.\star\]
for each $\var{a}\in\mathcal{V}_\var{a}$ and $\var{A}\in\mathcal{V}_\var{A}$. Note that $\lambda-$expressions for scalar variables can be omitted (and they actually were), since symbols $\sym{a}$ are constants. We use the notation $\theta(\cdot)$ in a more general way. For a given expression over program variables it always denotes the operation of replacing
each variable $\var{a}\in\mathcal{V}$ by the symbolic expression $\theta(\var{a})$. Moreover, we further extend the operation for formulae over program variables. We note that predicates containing $\star$ in some of its terms are immediately reduced to $\true$.

Let $\theta$ be a symbolic state, $\var{a}\in \mathcal{V}$ be a variable and $e$ be a symbolic expression. Then
$\theta[\var{a}\rightarrow e]$ is a symbolic state equal to $\theta$ except
for variable $\var{a}$, where $\theta[\var{a}\rightarrow e](\var{a})=e$.
Further, $\theta \langle e  \rangle$ denotes a symbolic expression derived from $e$ by simultaneously replacing all occurrences of each symbol $\sym{a}$ by symbolic
expression $\theta(\var{a})$. We also extend the notation $\theta \langle \cdot \rangle$ for formulae in natural way. We only note that predicates containing $\star$ in some its terms are immediately reduced to $\true$. Finally, we extend the notation $\theta \langle \cdot \rangle$ to symvolic states: $\theta \langle \theta' \rangle$ is a symbolic state satisfying
$\theta \langle \theta' \rangle (\var{a}) = \theta \langle \theta'(\var{a}) \rangle$ for each variable $\var{a}\in\mathcal{V}$. %Intuitively, the symbolic state $\theta_1\theta_2$
%represents an overall effect of symbolic execution of a code with effect
%$\theta_2$ followed by a code with effect $\theta_1$.

% \paragraph{Substitution into symbolic state} Let $\theta$ be a symbolic
% state of $P$ and $e, e'$ some symbolic expressions of $P$ of the same
% type. Then $\theta[e/e']$ is a symbolic state of $P$ such that for each
% variable $\var{a} \in \mathcal{V}(\mathcal{S}_P)$ we have
% $\theta[e/e'](\var{a}) = \theta(\var{a})[e/e']$. A symbolic state
% $\theta[e_1/e'_1, \ldots, e_n/e'_n]$ denotes simultaneous substitution of
% all pairs $e_i / e'_i$ into $\theta$.

For brevity of notation, we often use vector notation. Let
$\vec{u}=(u_1,\ldots,u_n)$ and $\vec{v}=(v_1,\ldots,v_n)$ be two vectors of
some symbolic expressions. We use $\vec{u}\le\vec{v}$ and
$\vec{u}<\vec{v}$ as abbreviations for the following formulae.
\begin{align*}
\vec{u}\le\vec{v} &~\equiv~ u_1\le v_1\wedge\ldots\wedge u_n\le v_n \\
\vec{u}<\vec{v} & 
~\equiv~\vec{u}\le\vec{v}\,\wedge\,\sum_{i = 1}^n{u_i}<\sum_{i = 1}^n{v_i}
\end{align*}

%}}}
%{{{ Algorithm

\section{Algorithm for Read-only Arrays}\label{sec:alg}

The idea of our algorithm is relatively simple. Given a program $P$, we
compute the set $B_P$ of all backbones of $P$, i.e.~all acyclic complete
paths. Then we compute an abstract path condition for each backbone. To
compute an abstract path condition for a backbone $\pi$, we perform a
standard symbolic execution of instructions along $\pi$ (i.e.~we gradually
construct a path condition $apc$ and we maintain a symbolic state $\theta$)
and whenever we visit an a entry node, we process the corresponding loop and then we add the resulting summary into the current path condition and symbolic state. 

Before we explain summary computation of loops, we need to define the
several terms. Let $v$ be an entry node of a loop $C$. An \emph{iteration}
is an arbitrary path of the form $v\alpha v$ such that $\alpha$ is a
(possibly empty) sequence of nodes in $C\smallsetminus\{v\}$. There is a
clear bijection between iterations and complete paths in the program
$P(C,v)$ induced by the loop. Hence, we do not distinguish between an
interation and the corresponding complete path. Let $\pi_1,\ldots,\pi_k$ be
all backbones in $P(C,v)$. We assosicate a fresh path counter $\kappa_i$ to
each backbone $\pi_i$. A \emph{looping path} is an arbitrary path over nodes
of $C$ leading from $v$ to $v$. Let $\beta$ be a looping path. Then $\beta$
can be written as $\beta=v\alpha_1v\alpha_2v\ldots v\alpha_nv$ where $n\ge
0$ and each $v\alpha_iv$ is an iteration. We define $\vec{\kappa}(\beta)$ as
vector $(c_1,c_2,\ldots,c_k)$, where each $c_j$ is the number of iterations
$v\alpha_iv$ in $\beta$ such that their backbone is~$\pi_j$.

To process a loop $C$ with entry node $v$ means to compute an iterated
symbolic state and looping condition for the loop. Let $\vec{\kappa}$ be a
vector of path counters firmly associated to backbones of the loop. On
intuitive level, \emph{Iterated symbolic state} $\theta^{\vec{\kappa}}$ is a
symbolic state that represents values of variables after arbitrary looping
path. The values are expressions that may contain path counters of
$\vec{\kappa}$. Further, \emph{looping condition} $\varphi^{\vec{\kappa}}$
is a formula generalizing all path conditions of all looping paths. The
formula $\varphi^{\vec{\kappa}}$ may contain path counters of
$\vec{\kappa}$. Formally, $\theta^{\vec{\kappa}}$ and
$\varphi^{\vec{\kappa}}$ have to satisfy the following condition: for each
path condition $pc$ and each symbolic memory $\theta$ produced a standard
symbolic execution along some looping path $\beta$, it holds that
\begin{itemize}
\item $pc\rightarrow\varphi^{\vec{\kappa}}[\vec{\kappa}/\vec{\kappa}(\beta)]$,
\item for each scalar variable $\var{a}$, either
  $\theta^{\vec{\kappa}}(\var{a})[\vec{\kappa}/\vec{\kappa}(\beta)]$
  contains $\star$, or $\theta(\var{a})=
  \theta^{\vec{\kappa}}(\var{a})[\vec{\kappa}/\vec{\kappa}(\beta)]$ is a
  valid formula.
\end{itemize}
Hence, an iterated symbolic state and a looping condition can be seen as an
abstraction (or an over-approximation) of all symbolic states and path conditions for all looping paths.

Now we return back to the symbolic execution of the backbone $\pi$. When we
have $\theta^{\vec{\kappa}}$ and $\varphi^{\vec{\kappa}}$, we update path
condition to $apc~\wedge~\theta\langle\varphi^{\vec{\kappa}}\rangle$, symbolic state to
$\theta\langle\theta^{\vec{\kappa}}\rangle$, and we continue with standard symbolic
execution along the backbone $\pi$. When the symbolic
execution of the backbone $\pi$ finishes, we set $apc$ to
$\exists\vec{\kappa}'(\vec{\kappa}'\ge\vec{0}~\wedge~apc)$, where
$\vec{\kappa}'$ is a vector of all path counters with free occurrences in
$apc$. The resulting formula $apc$ is a generalization of all (standard)
path conditions for all paths along the backbone $\pi$, as $pc\rightarrow apc$
holds for each such a standard path condition $pc$.

Let $apc_\pi$ be an abstract path condition for each backbone $\pi\in
B_P$. Then the necessary condition on input data to reach the target node
and thus the final result of our algorithm is the formula
$$\bigvee_{\pi\in B_P}apc_\pi.$$

\begin{algorithm}[!htb]
\LinesNotNumbered
\newcommand{\executeBB}{\texttt{executeBackbone}}
\caption{\executeBB\texttt{(}$\pi,P$\texttt{)}\label{alg:1}}
\KwIn{
\aargm{$\pi$}{a backbone of $P$}
\aargm{$P$}{a program}
}
\KwOut{
\aargm{$\theta$}{symbolic state}
\aargm{$apc$}{abstracted path condition}
}
\BlankLine 
\nl $\theta$ \aset $\theta_I$\; 
\nl $apc$ \aset $\true$\;
\nl Let $\pi$ has the form $v_0v_1\ldots v_n$\;
\nl \For{$i$\aset $1$ \KwTo $n$}{
\nl  \If{$\iota_P((v_{i-1},v_i))$ has the form $\texttt{assume}(\gamma)$}{
\nl     $apc$ \aset $apc\,\wedge\,\theta(\gamma)$\;
  }
\nl  \If{$\iota_P((v_{i-1},v_i))$ has the form $\var{v}\aset e$}{
\nl     $\theta$ \aset $\theta[\var{v}\rightarrow \theta(e)]$\;
  }
\nl  \If{$v_i$ is an entry node on $\pi$}{
\nl     Let $C$ be the loop at entry $v_i$ on the backbone $\pi$\;
\nl     Compute induced program $P(C,v_i)$\;
\nl     $(\theta^{\vec{\kappa}},\varphi^{\vec{\kappa}})$ \aset $\texttt{processLoop}(P(C,v_i))$\; \label{l:callProcessLoop} 
\nl     $apc$ \aset $apc\,\wedge\,\theta \langle \varphi^{\vec{\kappa}} \rangle$\;\label{alg:line-apc}
\nl     $\theta$ \aset $\theta\langle\theta^{\vec{\kappa}}\rangle$\;
  }
}
\nl $apc$ \aset $\exists\vec{\kappa}'(\vec{\kappa}'\ge\vec{0}~\wedge~apc)$\, where
$\vec{\kappa}'$ are all path counters\\\label{alg:line-apc2}
\hspace{26ex}~~with free occurences in $apc$\;
\nl \Return{$(\theta,apc)$}\;
\end{algorithm}

The computation of abstract path condition $apc$ for a given backbone $\pi$
is precisely formulated in Algorithm~\ref{alg:1}. On
Line~\ref{l:callProcessLoop}, the algorithm calls function
$\texttt{processLoop}(P(C,v_i))$ that returns an iterated symbolic state
$\theta^{\vec{\kappa}}$ and a looping condition $\varphi^{\vec{\kappa}}$ (i.e.~a summary) for
a loop $C$ at its entry node $v_i$ represented by an induced program $P(C,v_i)$. We assume that the path counters $\vec{\kappa}$
used in $\theta^{\vec{\kappa}}$ and $\varphi^{\vec{\kappa}}$ are fresh,
i.e.~they do not occur in current values of $\theta$ or $apc$.
% Note that the algorithm returns resulting abstract path condition $apc$ as
% well as the resulting symbolic state $\theta$. The returned symbolic
% states are needed by $\texttt{processLoop}$ algorithm which calls
% Algorithm~\ref{alg:1} on backbones of programs induced by loops.

In the rest of this section, we describe two versions of the
$\texttt{processLoop}$ procedure.

\subsection{Loop Processing: Lightweight Version}
\label{ssec:OverApproxLoop}

We are given a program $P'$ induced by a loop at some entry node. We compute the set of all
backbones $B_{P'}=\{\pi_1,\ldots,\pi_k\}$ and we run the function
\texttt{executeBackbone(}$\pi_i,P'$\texttt{)} on each backbone $\pi_i$. Let
$\theta_i$ and $apc_i$ be the returned symbolic state and abstract path
condition, respectively. Further, we assign a fresh path counter $\kappa_i$
to each backbone $\pi_i$. We set $\vec{\kappa}=(\kappa_1,\ldots,\kappa_k)$.

First, we compute an iterated symbolic state $\theta^{\vec{\kappa}}$. In
other words, for each scalar variable $\var{a}$ we construct a symbolic
expression over symbols and path counters of $\vec{\kappa}$ describing the
value of $\var{a}$ after arbitrary $\sum_{1\le m\le k}\kappa_i$ successive
executions of program $P'$ such that exactly $\kappa_i$ executions took
backbone $\pi_i$ for each $\pi_i\in B_{P'}$. In general, this is a very hard
task. To be on the safe side, we start with $\theta^{\vec{\kappa}}$ set to
$\theta_\star$ and we gradually improve its precision. More precisely, we
change the value of $\theta^{\vec{\kappa}}(\var{a})$ in one of the following
four cases:
\begin{enumerate}
\item For each backbone $\pi_i\in B_{P'}$, $\theta_i(\var{a})=\sym{a}$. In
  other words, the value of $\var{a}$ is not changed on any complete path in
  $P'$. This case is trivia. We set
  $\theta^{\vec{\kappa}}(\var{a})=\sym{a}$.
\item For each backbone $\pi_i\in B_{P'}$, either
  $\theta_i(\var{a})=\sym{a}$ or $\theta_i(\var{a})=\sym{a}+d_i$ for some
  symbolic expression $d_i$ such that $\theta^{\vec{\kappa}}\langle d_i \rangle$ contains
  neither $\star$ nor any path counters. Let us assume that the latter
  possibility holds for $\pi_1,\ldots,\pi_{k'}$ and the former one for
  $\pi_{k'+1},\ldots,\pi_k$. The condition on $\theta^{\vec{\kappa}}\langle d_i \rangle$
  guarantees that the value of $d_i$ is constant during all iterations over
  the loop. In this case, we set $\theta^{\vec{\kappa}}(\var{a})=
  \sym{a}+\sum_{1\le i\le k'}d_i\cdot\kappa_i$.
\item There exists a symbolic expression $d$ such that 
  $\theta^{\vec{\kappa}}\langle d \rangle$ contains neither $\star$ nor any path counters,
  and for each backbone $\pi_i\in B_{P'}$, either
  $\theta_i(\var{a})=\sym{a}$ or $\theta_i(\var{a})=d$. Let us assume that
  the latter possibility holds for $\pi_1,\ldots,\pi_{k'}$ and the former
  one for $\pi_{k'+1},\ldots,\pi_k$. In other words, the value of $\var{a}$
  is set to $d$ along each backbone $\pi_j$ for $1\le j\le
  k'$, while it is unchanged ony any other complete path. Hence, we set
  $\theta^{\vec{\kappa}}(\var{a})=\ite(\sum_{1\le j\le
    k'}\kappa_i>0,d,\sym{a})$.
\item For one backbone, say $\pi_i$, $\theta_i(\var{a})=d$ for some symbolic
  expression $d$ such that $\theta^{\vec{\kappa}}\langle d \rangle$ contains neither $\star$
  nor any path counters except $\kappa_i$. Further, for each backbone
  $\pi_j$ such that $i\neq j$, $\theta_j(\var{a})=\sym{a}$. That is, only
  the complete paths with backbone $\pi_i$ modify $\var{a}$ and they set it
  to a value independent on other path counters than $\kappa_i$. Note that
  if we assign $d$ to $\var{a}$ in the $\kappa_i$-th iteration along the
  complete paths with backbone $\pi_i$, then the actual assigned value of
  $d$ is the value after $\kappa_i-1$ iterations along the paths. Hence, we
  set $\theta^{\vec{\kappa}}(\var{a})=
  \ite(\kappa_i>0,\theta^{\vec{\kappa}}\langle d \rangle [\kappa_i/\kappa_i-1],\sym{a})$.
\end{enumerate}
Wa apply these rules repeatedly until no other precise value of
$\theta^{\vec{\kappa}}(\var{a})$ can be derived.

Computation of an looping condition $\varphi^{\vec{\kappa}}$ is
straightforward. The intuition has been already given in the introduction. We
set
\[
\setlength{\arraycolsep}{0pt}
\begin{array}{rl}
  \displaystyle
  \varphi^{\vec{\kappa}} \equiv \bigwedge_{i = 1}^k \forall \tau_i~ & \big(\,0
  \leq \tau_i < \kappa_i\,\rightarrow\\[-1ex]
  & ~~\rightarrow~ \exists \vec{\tau}_i~(\vec{0} \leq
  \vec{\tau}_i \leq \vec{\kappa}_i\,\wedge\,
  \theta^{\vec{\kappa}}\langle apc_i \rangle[\vec{\kappa}/\vec{\tau}])\big),
\end{array}
\]
where
\begin{align*}
  \vec{\tau}_i &= (\tau_1, \ldots, \tau_{i-1}, \tau_{i+1}, \ldots, \tau_k),\\
  \vec{\kappa}_i &= (\kappa_1, \ldots, \kappa_{i-1}, \kappa_{i+1}, \ldots,
  \kappa_k).
\end{align*}

Let us note that the program $P'$ induced by a loop can again contain
loops. Hence, a symbolic state $\theta_i$ and an abstract path condition
$apc_i$ for a backbone $\pi_i$ can contain path counters $\vec{\kappa}'$
corresponding to some loops inside $P'$. As the number of iterations of the
inner loop can be different in each iteration of the outer loop, the meaning
of $\vec{\kappa}'$ is different in each iteration of $P'$. Our construction
of $\theta^{\vec{\kappa}}$ handles this situation correctly: if some
$\theta_i(\var{a})$ contains a path counter of $\vec{\kappa}'$, then
$\theta^{\vec{\kappa}}(\var{a})=\star$. Further, the path counters of
$\vec{\kappa}'$ may occur in the looping condition $\varphi^{\vec{\kappa}}$,
but all their occurences are in subformulae of the form $apc_i$ talking
about a single iteration of $P'$, and they are bound there by an existential
quantifier.

\subsection{Improving Precision of Lightweight Version}

The loop processing procedure described in the previous subsection is
correct, but not very precise when program $P'$ contains nested loops. We
illustarte it on the following program.

\begin{minipage}{\textwidth}
\begin{verbatim}
for (i = 0; i < m; ++i) {
  j = i;
  while (j < n) {
    ++j;
  }
}
\end{verbatim}
\end{minipage}

\bigskip

The corresponding flowgraph is depicted in Figure~\ref{fig:ex2} (upper). The
program contains one backbone $\mathit{l_sal_t}$ with entry node $a$ and the
corresponding loop $C=\{a,b,c,d,e\}$. The induced program $P'=P(C,a)$
contains again one backbone $\mathit{abcea'}$ with entry node $c$ and the
corresponding loop $C'=\{c,d\}$. The induced programs $P'$ and $P''=P(C',c)$
are depicted in Figure~\ref{fig:ex2} (lower left and lower right respectively).

\begin{figure}[!htb]
  \centering
  % Define block styles
  \tikzstyle{start} = 
    [regular polygon,regular polygon sides=3,thick,draw,inner sep=1pt]
  \tikzstyle{loc} = [circle,thick,draw]
  \tikzstyle{pre} = [<-,shorten <=1pt,>=stealth',semithick]
  \tikzstyle{post} = [->,shorten <=1pt,>=stealth',semithick]
  \footnotesize
  \begin{tikzpicture}[node distance=1.4cm]
    \node [start,fill=black!10,regular polygon rotate=180] (ls) {$l_s$};
    \node [loc] (a) [below of=ls] {$a$}
      edge [pre] node [label=right:\texttt{i = 0}] {} (ls);
    \node [loc] (b) [below of=a] {$b$}
      edge [pre] node [label=right:\texttt{i < m}] {} (a);
    \node [loc] (c) [below of=b] {$c$}
      edge [pre] node [label=right:\texttt{j = i}] {} (b);
    \node [loc] (d) [below of=c] {$d$}
      edge [pre] node [label=left:\texttt{j < n}] {} (c)
      edge [post,bend right=90,looseness=2] node [label=right:\texttt{++j}] {} (c);
    \node [loc] (e) [left of=b,xshift=-2mm] {$e$}
      edge [pre,bend right=40] node [label=left:\texttt{j >= n~~}] {} (c)
      edge [post,bend left=40] node [label=left:\texttt{++i~~}] {} (a);
    \node [start,fill=black!10] (lt) [right of=b,xshift=2mm] {$l_t$}
      edge [pre,bend right=40] node [label=right:\texttt{~i >= m}] {} (a);

    \node [start,fill=black!10,regular polygon rotate=180,inner sep=2pt]
      (a1) [below of=d,xshift=-10mm] {$a$};
    \node [loc] (b1) [below of=a1] {$b$}
      edge [pre] node [label=right:\texttt{i < m}] {} (a1);
    \node [loc] (c1) [below of=b1] {$c$}
      edge [pre] node [label=right:\texttt{j = i}] {} (b1);
    \node [loc] (d1) [below of=c1] {$d$}
      edge [pre] node [label=left:\texttt{j < n}] {} (c1)
      edge [post,bend right=90,looseness=2] node [label=right:\texttt{++j}] {} (c1);
    \node [loc] (e1) [left of=c1,xshift=-5mm] {$e$}
      edge [pre] node [label=above:\texttt{j >= n}] {} (c1);
    \node [start,fill=black!10] (aa1) [below of=e1] {$a'$}
      edge [pre] node [label=left:\texttt{++i}] {} (e1);

    \node [start,fill=black!10,regular polygon rotate=180,inner sep=2pt]
      (c2) [right of=b1,xshift=20mm] {$c$};
    \node [loc] (d2) [below of=c2] {$d$}
      edge [pre] node [label=right:\texttt{j < n}] {} (c2);
    \node [start,fill=black!10] (cc2) [below of=d2] {$c'$}
      edge [pre] node [label=right:\texttt{++j}] {} (d2);
  \end{tikzpicture}
  \caption{Example of nested path counter dependency (upper). Program $P'$
    induced by loop $C=\{a,b,c,d,e\}$ with entry node $a$ (lower left).
    Program $P''$ induced by loop $C'=\{c,d\}$ with entry node $c$ (lower
    right).}
  \label{fig:ex2}
\end{figure}
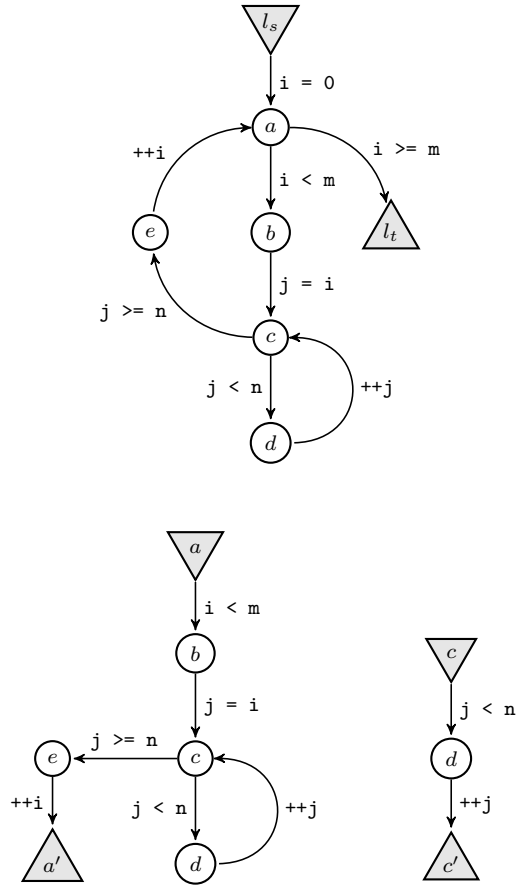

We can easily compute the iterated symbolic state $\theta^{\kappa'}$ and
looping condition $\varphi^{\kappa'}$ for $P''$:
\begin{align*}
  \theta^{\kappa'}(\var{j})&=\sym{j}+\kappa'
  \qquad \varphi^{\kappa'}\equiv\forall~\tau'(0\le\tau'<\kappa'~\rightarrow~\sym{j}+\tau'<\sym{n})\\
  \theta^{\kappa'}(\var{n})&=\sym{n}
\end{align*}
With this information, one can compute symbolic state $\theta'$ for backbone
$\mathit{abcea'}$ of $P'$. As there is only one backbone in $P'$, the
iterated symbolic state $\theta^\kappa$ can be computed directly from
$\theta'$.
\[
\setlength{\arraycolsep}{2pt}
\begin{array}{rclp{6ex}rcl}
  \theta'(\var{i})&=&\sym{i}+1           && \theta^\kappa(\var{i})&=&\sym{i}+\kappa\\
  \theta'(\var{j})&=&\sym{i}+\kappa'     && \theta^\kappa(\var{j})&=&\star\\
  \theta'(\var{m})&=&\sym{m}             && \theta^\kappa(\var{m})&=&\sym{m}\\
  \theta'(\var{n})&=&\sym{n}             && \theta^\kappa(\var{n})&=&\sym{n}
\end{array}
\]
In fact, the value of $\var{j}$ after one iteration of $P'$ can be expressed
without $\kappa'$ as $\theta'(\var{j})=\max(\sym{n},\sym{i})$. If we modify
$\theta'$ in this way, the algorithm presented in the previous section
computes more precise iterated symbolic state $\theta^\kappa$, namely it
returns
$$\theta^\kappa(\var{j})=\ite(\kappa>0,\max(\sym{n},\sym{i}+\kappa-1),\sym{j}).$$

The crucial step towards higher precision of iterated symbolic state is
detection of dependencies between path counters of an outer loop and path
counters of its nested loops. In the example, we would like to detect the
fact that in $(\kappa+1)$-st iteration of $P'$, the nested loop is iterated
$\kappa'=\max(0,\sym{n}-(\sym{i}+\kappa))$ times. In the next section we show how we detect these dependencies between path counters.

\subsection{Loop Processing: Heavyweight Version}

Intuitively, our heavyweight loop processing algorithm is looking for
\emph{linear dependency} of the \emph{sum of all path counters} of a nested
loop on path counters of the outer loop and on scalar program variables. We
are asking an SMT solver to infer dependecies from adjusted abstract path
conditions. If such a depenency is found, it is used to eliminate path
counters of nested loop in computation of iterated symbolic state
corresponding to the outer loop.

% Technically, we add a new artifical program variable $s$ for each nested
% loop with path counters $\vec{\kappa}'$ and we replace all occurences of
% $\sum \vec{\kappa}'$ in abstract path condition by $s$.

\begin{algorithm}[!htb]
%\LinesNotNumbered
\caption{\texttt{processLoop(}$P'$\texttt{)}\label{alg:2}}
\KwIn{
\aargm{$P'$}{an induced program of a loop at an entry node}
}
\KwOut{
\aargm{$\theta^{\vec{\kappa}}$}{iterated symbolic state}
\aargm{$\varphi^{\vec{\kappa}}$}{looping condition}
}
\BlankLine 
$\mathcal{V}_s$ \aset $\emptyset$\;
Compute backbones $B_{P'}=\{\pi_1,\pi_2,\ldots,\pi_k\}$\;
Let $\vec{\kappa}=(\kappa_1,\kappa_2,\ldots,\kappa_k)$ be fresh path counters\;
\ForEach{$\pi_i\in B_{P'}$}{
  $(\theta_i,apc_i)$ \aset \texttt{executeBackbone(}$\pi_i,P'$\textrm{)}\; 
  \ForEach{entry node $v_j$ on $\pi_i$}{
    Let $\vec{\kappa}_j$ be path counters of the loop entered by $v_j$\; 
    Add a fresh variable $\var{s}_{i,j}$ to $\mathcal{V}_s$\;
    Replace $\sum\vec{\kappa}_j$ in $\theta_i$ by $\sym{s}_{i,j}$
    and remove remaining $\vec{\kappa}_j$\;
    Construct a weakened looping condition $\sigma_j$ on $\sym{s}_{i,j}$\;
  }
}
Let $E$ be a set of all loop entry nodes along backbones in $B_{P'}$\;
\ForEach{entry node $v_{i,j} \in E$}{
  Construct a exit condition $\delta_{i,j}$ form formulae $apc_i$\;
}
Extend $\theta^{\vec{\kappa}}$ to $\mathcal{V}\cup\mathcal{V}_s$\;
$\theta^{\vec{\kappa}}$ \aset $\theta_\star$\;  
\Repeat{$\var{change} = \false$}{
  $\var{change} \aset \false$\; 
  \ForEach{$\var{a} \in \mathcal{V}$}{ \label{l:monotFirstLoop}
    $e$ \aset infer $\theta^{\vec{\kappa}}(\var{a})$ from $\theta_1,\theta_2,\ldots,\theta_k$ and $\theta^{\vec{\kappa}}$\;\label{alg:inf1}
    \If{$\theta^{\vec{\kappa}}(\var{a})=\star~\wedge~e\ne\star$}{
      $\theta^{\vec{\kappa}}(\var{a}) \aset e$\;
      $\var{change} \aset \true$\;
    }
  }
  \ForEach{$\var{s}_{i,j} \in \mathcal{V}_s$}{ \label{l:monotSecondLoop}
    $e$ \aset infer $\var{s}_{i,j}$ from $\sigma_{i,j} \wedge \delta_{i,j}$ and $\theta^{\vec{\kappa}}$\;\label{alg:inf2}
    \If{$\theta^{\vec{\kappa}}(\var{s}_{i,j})=\star~\wedge~e\ne\star$}{
      $\theta^{\vec{\kappa}}(\var{s}_{i,j}) \aset e$\;
      $\var{change} \aset \true$\; \label{l:monotoneFnEnd}
    }
  }
}
$\theta^{\vec{\kappa}} \aset \theta^{\vec{\kappa}}|_{\mathcal{V}}$\;
Construct $\varphi^{\vec{\kappa}}$ from $apc_1,\ldots,apc_k$ and $\theta^{\vec{\kappa}}$\;
\Return{$(\theta^{\vec{\kappa}},\varphi^{\vec{\kappa}}$)}\;
\end{algorithm}

The heavyweight loop processing procedure is given in
Algorithm~\ref{alg:2}. The algorithm works with the set of artificial
program variables $\mathcal{V}_s$, which is empty at the beginning. The
algorithm starts similarly as the lightweight one: it computes all backbones
of $P'$ and for each backbone $\pi_i$ is computes the corresponding symbolic
state $\theta_i$ and abstract path condition $apc_i$. Further, for each loop
on the backbone $\pi_i$ it performs the following three steps, where
$\vec{\kappa}_j$ are path counters of the loop.
\begin{itemize}
\item A fresh artificial variable $\var{s}_{i,j}$ is added to
  $\mathcal{V}_s$. In each iteration with backbone $\pi_i$, the variable
  $\var{s}_{i,j}$ represents the sum of path counters $\vec{\kappa}_j$.
\item We change $\theta_i$ to
  $\theta_i[\sum \vec{\kappa}_j/\sym{s}_{i,j}][\vec{\kappa}_j/\vec{\star}]$.
  Hence, we replace each sum of all path counters $\vec{\kappa}_j$ by
  $\sym{s}_{i,j}$ and we replace all other occurences of path counters
  $\vec{\kappa}_j$ by $\star$.
\item Since we are only interested in the sum $\sum\vec{\kappa}_j$, we can weaken the looping condition of the inner loop. We denote this weakened formula as $\sigma_{i,j}$ and compute it as follows. Let $apc'_1, \ldots,$ $apc'_{k'}$ be abstracted path conditions for
  backbones of the inner loop and let $\theta^{\vec{\kappa}_j}$ be the
  iterated symbolic state computed for the inner loop. Then we set
  \begin{align*}
    \sigma_{i,j}\,\equiv~&0\le\sym{s}_{i,j}{-}1~\rightarrow\\ & \theta^{\vec{\kappa}_j} \langle apc'_1 \vee \cdots \vee apc'_{k'} \rangle [\sum\vec{\kappa}_j/\sym{s}_{i,j}{-}1][\vec{\kappa}_j/\vec{\star}].
  \end{align*}
\end{itemize}
Next we compute exit conditions $\delta_{i,j}$ from inner loops at loop entries $v_{i,j}$. For each exit node $x$ from an inner loop at entry node $v_{i,j}$, there is a backbone $\pi_l$, $l \leq k$, of the form $\pi_l = \ldots v_{i,j} \alpha x \ldots$. To leave the loop, all the conditions along the path $\alpha x$ must by satisfiable. Therefore, if we denote by $apc_{v_{i,j},x}(\alpha x)$ conjunction of these conditions and $x_1, \ldots, x_r$ are all exits from the inner loop at entry $v_{i,j}$, then we can express $\delta_{i,j}$ as a formula:
\begin{align*}
\delta_{i,j} \equiv (apc_{v_{i,j},x_1}&(\alpha_1 x_1) \vee \cdots \vee apc_{v_{i,j},x_r}(\alpha_r x_r))\\&[\sum\vec{\kappa}_j/\sym{s}_{i,j}{-}1][\vec{\kappa}_j/\vec{\star}]
\end{align*}

In the second half of the algorithm, we extend $\theta^{\vec{\kappa}}$ to
the artificial program variables and we compute iterated symbolic state
$\theta^{\vec{\kappa}}$. We alternately try to infer more precise
information for standard and artificial program variables. The inference on
line~\ref{alg:inf1} employs the four conditions formulated in
Subsection~\ref{ssec:OverApproxLoop}, while the inference on
line~\ref{alg:inf2} executes an SMT solver. The solver decides whether there
exists a linear expression over path counters $\vec{\kappa}$ and constant
symbols corresponding to scalar program variables equivalent to
$\sym{s}_{i,j}$ for each $\sym{s}_{i,j}$ that satisfies the necessary
condition given by $\sigma_{i,j} \wedge \delta_{i,j}$. Hence, an SMT solver asked for
satisfieability of the formula
\begin{align*}
  \forall\,\vec{\sym{a}},\vec{\kappa},\sym{s}_\gamma~\Big(& (\vec{\kappa} \ge
  \vec{0} \wedge \sym{s}_\gamma \ge
  0 \wedge \theta^{\vec{\kappa}} \langle \sigma_{i,j} \wedge \delta_{i,j} \rangle) ~\rightarrow\\
  & ~\rightarrow~~ \sym{s}_\gamma=\max\{0,(\vec{\kappa}\cdot
  M+\vec{w})\cdot(\vec{\sym{a}},1)^T\}\Big)\textrm{,}
\end{align*}
where $M$ is a matrix and $\vec{w}$ is a vector of constant symbols and of
an appropriate sizes, $\vec{\sym{a}}$ is a vector of variables containing a
variable $\sym{a}$ for each scalar program variable $\var{a}$, and
$(\var{a},1)$ is the same vector prolonged with constant 1. Note that it is
possible to formulate stronger necessary conditions on $\sym{s}_{i,j}$ than
$\rho_{i,j}$. A stronger condition can lead to more discovered
dependencies. One can also look for more complex dependencies (for example
dependencies involving arrays). We chose a simple and relatively weak
conditions $\rho_{i,j}$ to get quick reactions of an SMT solver.

At the end of algorithm, we restrict the iterated symbolic state obtained
$\theta^{\vec{\kappa}}$ back to standard program variables. Finally, we
compute a looping condition in the way described in
Subsection~\ref{ssec:OverApproxLoop}.

%}}}
%{{{ Algorithm - arrays

\section{Extension for Array-manipulating Programs}\label{sec:alg2}

This section sketches necessary steps to extension of our algorithm to
programs that modify arrays. 

First, we extend our instruction set with an assignment instruction of the
form $\texttt{A[}e_1,e_2,\ldots,e_n\texttt{]}\aset e$, where
$e,e_1,e_2,\ldots,e_n$ are program expressions of integer type and $n\ge 1$
is the arity of $\texttt{A}$. Further, we have to define symbolic
expressions of types $\texttt{Int}^k\rightarrow\texttt{Int}$ for every arity
$k$. For expressions of suhc a type, we use the notation
$\lambda\chi_1\chi_2\ldots\chi_k.e$ or
$\lambda\chi_1\chi_2\ldots\chi_k.e(\chi_1,\chi_2,\ldots,\chi_k)$ if we want
to emphasize that $e$ is a function symbol of type
$\texttt{Int}^k\rightarrow\texttt{Int}$. We often use vector notation
$\lambda\vec{\chi}.e(\vec{\chi})$ instead of
$\lambda\chi_1\chi_2\ldots\chi_k.e(\chi_1,chi_2,\ldots,\chi_k)$, where $k$
is determined by a context. A symbolic execution has to be extended as well
in order to handle assignment instructions modifying an array.

The most interesting of the extension for full array support are the rules
that allow us to compute values of array variables in and iterated symbolic
state. We mention two rules we have designed for arrays.

Assume that we are given a program $P'$ with backbones $\pi_1,\ldots,\pi_k$
and symbolic state $\theta_i$ and abstract path condition $apc_i$ for each
$\pi_i$. The iterated symbolic value of an array variable $\var{A}$ can be
precisely computed if some of the following two cases happen.
\begin{itemize}
\item If $\theta_i(\var{A})=\lambda\vec{\chi}.\sym{A}(\vec{\chi})$ for
  all backbones $\pi_i$, then the array is not changed along any complete
  path in $P'$. Hence we set
  $\theta^{\vec{\kappa}}(\var{A})=\lambda\vec{\chi}.\sym{A}(\vec{\chi})$.
\item if there exists one backbone, say $\pi_1$ such that
  \begin{align*}
    \theta_1(\var{A})=
    \lambda\vec{\chi}.\ite(&\varphi_1,t_1(\vec{\chi}),\ite(\varphi_2,t_2(\vec{\chi}),
    \ldots\\
    &\ldots,\ite(\varphi_n,t_n(\vec{\chi}),\sym{A}(\vec{\chi}))\ldots))
  \end{align*}
  and $\theta_i(\var{A})=\lambda\vec{\chi}.\sym{A}(\vec{\chi})$ for all
  other backbones $\pi_1$, then we set 
  \begin{align*}
    \theta^{\vec{\kappa}}(\var{A})=
    \lambda\vec{\chi}.&\ite(\kappa_1=0,\sym{A}(\vec{\chi}),\ite(\varphi_1,t_1(\vec{\chi}),\\
    & \ite(\varphi_2,t_2(\vec{\chi}),\ldots,\ite(\varphi_n,t_n(\vec{\chi}),\sym{A}(\vec{\chi}))\ldots))).
  \end{align*}
\end{itemize}

Other rules are written in similar style. 

%}}}
%{{{ Soundness and Incompleteness

\section{Soundness and Incompleteness}

In this section we formulate and prove soundness and incompleteness theorems for our algorithm.

\begin{thm}[Soundness]
Let $\mathit{apc}$ be the necessary condition computed by our algorithm for a given target program location. If $\mathit{apc}$ is not satisfiable, then the target location is not reachable in that program.
\end{thm}
\begin{proof}[Informal proof]
We build any looping condition $\varphi^{\vec{\kappa}}$ such that it is implied by all path conditions of an analysed loop. And each formula $\mathit{apc}_\pi$ constructed in Algorithm~\ref{alg:1} collects all the predicated along passed backbone $\pi$ and it also collects looping conditions at loop entry nodes along the backbone. Therefore, $\mathit{apc}_\pi$ must be implied by any path condition of any symbolic execution along $\pi$. We compute final $\mathit{apc}$ as a disjunction of formulae $\mathit{apc}_\pi$ for all backbones. Since any program path leading to the target location must follow some backbone (with possible temporary escapes into loops along the backbone), its path condition exists (i.e.~it is satisfiable formula) only if $\mathit{apc}$ is satisfiable.
%\label{alg:execBB:addLoopCond}
%$\mathit{apc}$ is 
\end{proof}

\begin{thm}[Incompleteness]
There is a program and an unreachable target location in it for which the formula $\mathit{apc}$ computed by our algorithm is satisfiable.
\end{thm}
\begin{proof}
Let us consider the following C code:
\begin{verbatim}
   int i = 1;
   while (i < 3) {
      if (i == 2)
         i = 1;
      else
         i = 2;
   }
\end{verbatim}
The loop never terminates. Therefore, a program location below it is not reachable. But $\mathit{apc}$ computed for that location is equal to $\true$, since variable $\var{i}$ does not follow a monotone progression.
\end{proof}

%}}}
%{{{ Quantifiers

\section{Dealing with Quantifiers}\label{sec:quantifiers}

We can ask an SMT solver whether a computed necessary condition $\mathit{apc}$ is satisfiable or not. And if it is, we may further ask for some its model. As we will see in Section~\ref{sec:apps} such queries to a solver should be fast. Unfortunately, our experience with solvers shows that presence of quantifiers in $\mathit{apc}$ usually causes performance issues. Although SMT technology evolves quickly, we show in this section how to overcome this issue now by unfolding universally quantified formulae the looping conditions $\varphi^{\vec{\kappa}}$ are made of.

Universally quantified variables $\tau_i$ in formulae $\varphi^{\vec{\kappa}}$ are always restricted from above by path counters $\kappa_i$ counting iterations of backbones $\pi_i'$ of analysed loop. Let us choose some upper limits $K_i > 0$ for the path counters $\kappa_i$. Since each $\tau_i$ ranges over a finite set of integers $\{ 0, \ldots, K_i - 1 \}$ now, we can unfold each universally quantified formula in $\varphi^{\vec{\kappa}}$ for each possible value of $\tau_i$. Having eliminated the universal quantification, we can also eliminate existential quantification of all $\kappa_i$ and all $\vec{\tau_i}$ in $\varphi^{\vec{\kappa}}$ and whole $\mathit{apc}$ by redefining them as uninterpreted integer constants. Let us see an unfolded necessary condition $\mathit{apc}$, denoted by $\mathit{apc}^{\vec{K}}$, of our running example, when we choose upper limits $\vec{K} = (K_1, K_2)$ for the path counters $\vec{\kappa} = (\kappa_1, \kappa_2)$:
\begin{tabbing}
~
$\mathit{apc}^{\vec{K}} \equiv~$ \= $0 \leq \kappa_1 \wedge 0 \leq \kappa_2~\wedge$ \\
\> $\bigwedge_{i = 0}^{K_1} (0 \leq i < \kappa_1 \rightarrow (0 \leq \tau_{2,i} \leq \kappa_2~\wedge$ \\
\>\hspace{1.5cm} $i + \tau_{2,i} < \sym{n} \wedge \sym{A}(i + \tau_{2,i}) = 1))~\wedge$ \\
\> $\bigwedge_{i = 0}^{K_2} (0 \leq i < \kappa_2 \rightarrow (0 \leq \tau_{1,i} \leq \kappa_1~\wedge$\\
\>\hspace{1.5cm} $\tau_{1,i} + i < \sym{n} \wedge \sym{A}(\tau_{1,i} + i) \neq 1))~\wedge$ \\
\> $\kappa_1 + \kappa_2 \geq \sym{n}~\wedge~\kappa_1 + 3 > 12$,
\end{tabbing}
where $\kappa_1, \kappa_2, \tau_{1,0}, \ldots, \tau_{1,K_2}, \tau_{2,0}, \ldots, \tau_{2,K_1}$ are uninterpreted integer constants.

For any $\vec{K}$ the formula $\mathit{apc}^{\vec{K}}$ represents wakened $\mathit{apc}$. Higher values we choose, then we get closer to the precision of $\mathit{apc}$. In practice we must choose moderate values $\vec{K}$, since the unfolding process makes $\mathit{apc}^{\vec{K}}$ much longer then $\mathit{apc}$.

%Let $\vec{K_0}, \vec{K_1}, \ldots$ be a sequence of guessed upper limits for the path counters, where for all $i < j$ we have $\vec{K_i} < \vec{K_j}$. We use the unfolding of $\mathit{apc}$ such that we ask an SMT solver for satisfiability of formulae $\mathit{apc}^{\vec{K_0}}, \mathit{apc}^{\vec{K_1}}, \ldots$ one by one, until we find the first satisfiable formula or we exceed some timeout.

In some cases an SMT solver is able to quickly decide satisfiability of $\mathit{apc}$. Therefore, we ask the solver for satisfiability of $\mathit{apc}$ in parallel with the unfolding procedure described above. And there is a common timeout for both queries. We take the fastest answer. In case both queries exceeds the timeout, the condition $\mathit{apc}$ cannot help a tool to cover given target location.

%}}}
%{{{ Applications

\section{Integration into Tools}\label{sec:apps}

%Tests generation tools based on symbolic execution
Tools based on symbolic execution typically explore program paths iteratively. At each iteration there is a set of program locations $\{ v_1, \ldots, v_k \}$, from which the symbolic execution may continue further. At the beginning the set contains only program entry location. In each iteration of the symbolic execution the set is updated such that actions of program edges going out from \emph{some} locations $v_i$ are symbolically executed. Different tools use different systematic and heuristic strategies for selecting locations $v_i$ to be processed in the current iteration. It is also important to note that for each $v_i$ there is available an actual path condition $\varphi_i$ capturing already taken symbolic execution from the entry location up to $v_i$.

When a tool detects difficulties in some iteration to cover a particular program location, then using $\mathit{apc}$ it can restrict selection from the whole set $\{ v_1, \ldots,$ $v_k \}$ to only those locations $v_i$, for which a formula $\varphi_i \wedge \mathit{apc}$ is satisfiable. In other words, if for some $v_i$ the formula $\varphi_i \wedge \mathit{apc}$ is \emph{not} satisfiable, then we are guaranteed there is no real path from $v_i$ to the target location. And therefore, $v_i$ can safely be removed from the consideration.

Tools like \Sage, \Pex or \Cute combine symbolic execution with concrete one. Let us assume that a location $v_i$, for which the formula $\varphi_i \wedge \mathit{apc}$ is satisfiable, was selected in a current iteration. These tools require a concrete input to the program to proceed further from $v_i$. Such an input can directly be extracted from any model of the formula $\varphi_i \wedge \mathit{apc}$.

%}}}
%{{{ Experimental results

\section{Experimental Results}\label{sec:experiments}

We implemented the algorithm in an experimental program, which we call \APC. We also prepared a small set of benchmark programs mostly taken from other papers. In each benchmark we marked a single location as the target one. All the benchmarks have a huge number of paths, so it is difficult to reach the target. We run \Pex and \APC on the benchmarks and we measured times till the target locations were reached. This measurement is obviously unfair from \Pex perspective, since its task is to cover an analysed benchmark by tests and not to reach a single particular location in it. Therefore, we clarify the right meaning of the measurement now.

Our only goal here is to show, that \Pex could benefit from our algorithm. Typical scenario when running \Pex on a benchmark is that all the code except the target location is covered in few seconds (typically up to three). Then \Pex keeps searching space of program paths for a longer time without covering the target location. This is exactly the situation when our algorithm should be activated. We of course do not know the exact moment, when \Pex would activate it. Therefore, we can only provide running times of our algorithm as it was activated at the beginning of the analysis.
% If we want to show that \Pex could benefit from the algorithm, then \APC must be significantly faster then \Pex in more benchmarks.

Before we present the results, we discuss the benchmarks. Benchmark HWM checks whether an input string contains four substrings \texttt{Hello}, \texttt{world}, \texttt{at} and \texttt{Microsoft!}. It does not matter at which position and in which order the words occur in the string. The target location can be reached only when all the words are presented in the string. This benchmark was introduced in~\cite{AGT08}. The benchmark consists of four loops in a sequence, where each loop searches for a single of the four words mentioned above. Each loop checks for an occurrence of a related word at each position in the input string starting from the beginning. Benchmark HWM is the most complicated one from our set of benchmarks. We also took its two lightened versions presented in~\cite{OT11}: Benchmark HW consists of two loops searching the input string for the first two words above. And benchmark Hello searches only for the first one.% It is interesting to observe performance of \Pex and \APC on these three benchmarks, as their complexity grows exponentially from Hello to HWM.

Benchmark MatrIR scans upper triangle of an input matrix. The matrix can be of any rank bigger then $20 \times 20$. In each row we count a number of elements inside a fixed range $(10,100)$. When a count for any row exceeds a fixed limit $15$, then the target location is reached.

Benchmarks OneLoop and TwoLoops originate from~\cite{OT11}. They are designed such that their target locations are not reachable. Both benchmarks contain a loop in which the variable \var{i} (initially set to 0) is increased by 4 in each iteration. The target location is then guarded by an assertion \texttt{i==15} in OneLoop benchmark and by a loop\texttt{~while (i != j + 7) j += 2~}in the second one. We note that \var{j} is initialized to $0$ before the loop.

The last benchmark WinDriver comes from a practice and we took it from~\cite{GLE09}. It is a part of a Windows driver processing a stream of network packets. It reads an input stream and decomposes it into a two dimensional array of packets. A position in the array where the data from the stream are copied into are encoded in the input stream itself. We marked the target location as a failure branch of a consistency check of the filled in array. It was discussed in the paper~\cite{GLE09} the consistency check can indeed be broken.

\begin{table*}[!htb]
	\begin{center}
    \begin{tabular}{c||c|c|c|c|c}
      & \Pex & \multicolumn{4}{c}{\APC} \\
      \cline{2-6}
      \textbf{Benchmark} & \textbf{Total} & \textbf{Total} & Bld $\mathit{apc}$ & Unf/SMT $\mathit{apc}^{\vec{K}}$ & SMT $\mathit{apc}$ \\
      \hline
      Hello       & 5.257 & 0.181 & 0.021 & 0.290 / S 0.060 & S 0.160 \\
      HW          & 25.05 & 0.941 & 0.073 & 0.698 / S 0.170 & S 13.84 \\
      HWM         & T/O   & 4.660 & 1.715 & 2.135 / S 0.810 & X M/O~ \\
      MatrIR      & 95.00 & 0.035 & 0.015 & 0.491 / S 70.80 & S 0.020 \\
      WinDriver   & 28.39 & 0.627 & 0.178 & 0.369 / S 0.080 & X 4.860 \\
      \hline
      OneLoop     & 134.0 & 0.003 & 0.001 & 0.001 / U 0.001 & U 0.010 \\
      TwoLoops    & 64.00 & 0.003 & 0.002 & 0.004 / U 0.010 & U 0.001
      %\multicolumn{6}{c}{} \\
    \end{tabular}
%    {\tiny
%      \begin{tabular}{l}
%        - Intel\textsuperscript{\textregistered} Core$^{\mathtt{TM}}$ i7 CPU 920 @ 2.67GHz 2.67GHz, 6GB RAM, Windows 7 Professional 64-bit \\
%        - MS \Pex 0.92.50603.1, MS Moles 1.0.0.0, MS Visual Studio 2008, MS .NET Framework v3.5 SP1 \\
%        - MS Z3 SMT solver v3.2, and boost v1.42.0.
%      \end{tabular}
%    }
	\end{center}
  \caption{Running times of \Pex and \APC on benchmarks.}
  \label{tab:experiments}
\end{table*}

The experimental results are depicted in Table~\ref{tab:experiments}. They show running times in seconds of \Pex and \APC on the benchmarks. We did all the measurements on a single common desktop computer\footnote{Intel\textsuperscript{\textregistered} Core$^{\mathtt{TM}}$ i7 CPU 920 @ 2.67GHz 2.67GHz, 6GB RAM, Windows 7 Professional 64-bit, MS \Pex 0.92.50603.1, MS Moles 1.0.0.0, MS Visual Studio 2008, MS .NET Framework v3.5 SP1, MS \Z SMT solver v3.2, and boost v1.42.0.}. The mark T/O in \Pex column indicates that it failed to reach the target location within an hour. For \APC we provide the total running times and also time profiles of different paths of the computation. In sub-column 'Bld $\mathit{apc}$' there are times required to build the necessary condition $\mathit{apc}$. In sub-column 'Unf/SMT $\mathit{apc}^{\vec{K}}$' there are two times for each benchmark. The first number identifies a time spent by unfolding the formula $\mathit{apc}$ into $\mathit{apc}^{\vec{K}}$. We use a fixed number 25 for all the counters and benchmarks. The second number represent a time spent by \Z SMT solver~\cite{Z3} to decide satisfiability of the unfolded formula $\mathit{apc}^{\vec{K}}$. Characters in front of these times identify results of the queries: S for satisfiable, U for unsatisfiable and X for unknown. And the last sub-column 'SMT $\mathit{apc}$' contains running times of \Z SMT solver directly on formulae $\mathit{apc}$. The mark M/O means that \Z went out of memory. As we explained in Section~\ref{sec:quantifiers} the construction and satisfiability checking of $\mathit{apc}^{\vec{K}}$ runs in parallel with satisfiability checking of $\mathit{apc}$. Therefore, we take the minimum of the times to compute the total runing time of \APC.

%}}}
%{{{ Related work

\section{Related Work}\label{sec:related}

Early work on symbolic execution~\cite{Kin76,BEL75,How77} showed its effectiveness in test generation. %Each generated test is guaranteed to explore a new yet unseen behaviour of the analysed program.
King further showed that symbolic execution can bring more automation into Floyd's inductive proving method~\cite{Kin76,Floyd67}. Nevertheless, loops as the source of the path explosion problem were not in the center of interest.

More recent approaches dealt mostly with limitations of SMT solvers and the environment problem by combining the symbolic execution with the concrete one~\cite{PKS05,AGT08,SMA05,GLM08:active_props,G07,GLM08:fuzzing,GKL08,TdH08,GLM08:fuzzing,PRV11}. Although practical usability of the symbolic execution improved, these approaches still suffer from the path explosion problem.
%The approaches we discuss next focus primarily on the path explosion problem.
An interesting idea
% dealing with the path explosion problem 
is to combine the symbolic execution with a complementary technique~\cite{GNRT10,GMR09,Beckmanetal08,NRTT09,Gulavanietal06}. Complementary techniques typically perform differently on different parts of the analysed program. Therefore, an information exchange between the techniques leads to a mutual improvement of their performance.
There are also techniques based on saving of already observed program behaviour and early terminating those executions, whose further progress will not explore a new one~\cite{BCE08,Cadar08,CDE08}.
Compositional approaches are typically based on computation of function summaries~\cite{G07,AGT08}. A function summary often consists of pre and post condition. Preconditions identify paths through the function and postconditions capture effects of the function along those paths. Reusing these summaries at call sites typically leads to an interesting performance improvement. In addition the summaries may insert additional symbolic values into the path condition which causes another improvement. % of performance.
And there are also techniques partitioning program paths into separate classes according to similarities in program states~\cite{QNR11,SH10}. Values of output variables of a program or function are typically considered as a partitioning criteria. %Such a technique then analyses only one program path per class.
A search strategy Fitnex~\cite{XTHS09} implemented in \Pex~\cite{TdH08} uses state-dependent fitness values computed through a fitness function to guide a path exploration. The function measures how close an already discovered feasible path is to a particular target location (to be covered by a test). The fitness function computes the fitness value for each occurrence of a predicate related to a chosen program branching along the path. The minimum value is the resulting one. There are also orthogonal approaches dealing with the path explosion problem by introducing some assumptions about program input. There are, for example, specialized techniques for programs manipulating strings~\cite{BTV09,XGM08}, and techniques reducing input space by a given grammar~\cite{GKL08,SPmCS09}.

Although the techniques above showed performance improvements when dealing with the path explosion problem, they do not focus directly on loops.% Next we discuss approaches focusing on them.%, where program loops are in the center of interest.
The LESE~\cite{SPmCS09} approach introduces symbolic variables for the number of times each loop was executed and links these with features of a known input grammar such as variable-length or repeating fields. This allows the symbolic constraints to cover a class of paths that includes different number of loop iterations, expressing loop-dependent program values in terms of %properties of 
the input.
A technique presented in~\cite{GL11} analyses loops on-the-fly, i.e.~during simultaneous concrete and symbolic execution of a program for a concrete input. The loop analysis infers inductive variables. A variable is inductive if it is modified by a constant value in each loop iteration. These variables are used to build loop summaries expressed in a form of pre a post conditions. The summaries are derived from the partial loop invariants synthesized dynamically using pattern matching rules on the loop guards and induction variables.
In our previous work~\cite{OT11} we introduced an algorithm sharing the same goal as one presented here. Nevertheless, in~\cite{OT11} we transform an analysed program into chains and we do the remaining analysis there. For each chain with sub-chains we build a constraint system serving as an oracle for steering the symbolic execution in the path space towards the target location.

%}}}
%{{{ Conclusion

\section{Conclusion}\label{sec:conclusion}

We presented algorithm computing for a given target program location the necessary condition $\mathit{apc}$ representing an over-approximated set of real program paths leading to the target. We proposed the use of $\mathit{apc}$ in tests generation tools based on symbolic execution. Having $\mathit{apc}$ such a tool can cover the target location faster by exploring only program paths in the over-approximated set. We also showed that $\mathit{apc}$ can be used in the tools very easily and naturally. And we finally showed by the experimental results that \Pex could benefit from our algorithm.

%}}}
%{{{ Acknowledgements 

% \acks
% Acknowledgments, if needed.

%}}}

\bibliographystyle{plain}
\bibliography{apc}

% The bibliography should be embedded for final submission.
% \begin{thebibliography}{}
% \softraggedright
% \bibitem[Smith et~al.(2009)Smith, Jones]{smith02}
% P. Q. Smith, and X. Y. Jones. ...reference text...
% \end{thebibliography}

\appendix
\newpage

%{{{ Benchmarks Listings

\section{Listing of Benchmarks in C\#}

\begin{verbatim}
public static void Hello(string A)
{
    string H = "Hello\0"; int h = 0;
    for (int i = 0; A[i] != 0; ++i)
    {
        int j = i, k = 0;
        while (H[k] != 0 && A[j] != 0 &&
               A[j] == H[k])
        { ++j; ++k; }
        if (H[k] == 0) { h = 1; break; }
        if (A[j] == 0) break;
    }
    if (h == 1)
        throw new Exception("Hello - reached!!");
}
\end{verbatim}

\begin{verbatim}
public static void HW(string A)
{
    string H = "Hello\0"; int h = 0;
    for (int i = 0; A[i] != 0; ++i)
    {
        int j = i, k = 0;
        while (H[k] != 0 && A[j] != 0 &&
        		   A[j] == H[k])
        { ++j; ++k; }
        if (H[k] == 0) { h = 1; break; }
        if (A[j] == 0) break;
    }
    string W = "World\0"; int w = 0;
    for (int i = 0; A[i] != 0; ++i)
    {
        int j = i, k = 0;
        while (W[k] != 0 && A[j] != 0 &&
               A[j] == W[k])
        { ++j; ++k; }
        if (W[k] == 0) { w = 1; break; }
        if (A[j] == 0) break;
    }
    if (h == 1 && w == 1)
        throw new Exception("HW - reached!!");
}
\end{verbatim}

\begin{verbatim}
public static void HWM(string A)
{
    string H = "Hello\0"; int h = 0;
    for (int i = 0; A[i] != 0; ++i)
    {
        int j = i, k = 0;
        while (H[k] != 0 && A[j] != 0 &&
        		   A[j] == H[k])
        { ++j; ++k; }
        if (H[k] == 0) { h = 1; break; }
        if (A[j] == 0) break;
    }
    string W = "World\0"; int w = 0;
    for (int i = 0; A[i] != 0; ++i)
    {
        int j = i, k = 0;
        while (W[k] != 0 && A[j] != 0 &&
               A[j] == W[k])
        { ++j; ++k; }
        if (W[k] == 0) { w = 1; break; }
        if (A[j] == 0) break;
    }
    string T = "At\0"; int t = 0;
    for (int i = 0; A[i] != 0; ++i)
    {
        int j = i, k = 0;
        while (T[k] != 0 && A[j] != 0 &&
               A[j] == T[k])
        { ++j; ++k; }
        if (T[k] == 0) { t = 1; break; }
        if (A[j] == 0) break;
    }
    string M = "Microsoft!\0"; int m = 0;
    for (int i = 0; A[i] != 0; ++i)
    {
        int j = i, k = 0;
        while (M[k] != 0 && A[j] != 0 &&
               A[j] == M[k])
        { ++j; ++k; }
        if (M[k] == 0) { m = 1; break; }
        if (A[j] == 0) break;
    }
    if (h == 1 && w == 1 && t == 1 && m == 1)
        throw new Exception("HWM - reached!!");
}
\end{verbatim}

\begin{verbatim}
public static void MatrIR(int[,] A, int m, int n)
{
    int w = 0;
    for (int i = 0; i < m; ++i)
    {
        int k = 0;
        for (int j = i; j < n; ++j)
            if (A[i, j] > 10 && A[i, j] < 100)
                ++k;
        if (k > 15)
        {
            w = 1;
            break;
        }
    }
    if (m > 20 && n > 20 && w == 1)
        throw new Exception("MatrIR - reached!!");
}
\end{verbatim}

\begin{verbatim}
public static void OneLoop(int n)
{
    int i = 0;
    while (i < n) i += 4;
    if (i == 15)
        throw new Exception("OneLoop - reached!!");
}
\end{verbatim}

\begin{verbatim}
public static void TwoLoops(int n)
{
    int i = 0, j = 0;
    while (i < n) i += 4;
    while (i != j + 7) j += 2;
    throw new Exception("TwoLoops - reached!!");
}
\end{verbatim}

\begin{verbatim}
public static void WinDriver(int[,] multi_array,
    int[] buffer, int MAX_PACKET, int PACKET_SIZE)
{
    for (int i = 0; i < MAX_PACKET; ++i)
        for (int j = 0; j < PACKET_SIZE; j++)
            multi_array[i,j] = 0;
    int number_of_packets;
    int packet_id;
    number_of_packets = (int)buffer[0];
    if ((number_of_packets > MAX_PACKET) ||
        (number_of_packets < 0))
        return;
    for (int i = 0; i < number_of_packets; i++)
    {
        packet_id =
           (int)buffer[(i * (PACKET_SIZE + 1)) + 1];
        if ((packet_id >= MAX_PACKET) ||
            (packet_id < 0))
            return;
        for (int j = 0; j < PACKET_SIZE; j++)
            multi_array[packet_id,j] =
               buffer[(i * (PACKET_SIZE + 1)) + j + 2];
    }
    if ((number_of_packets < MAX_PACKET) &&
        (multi_array[number_of_packets,0] != 0) &&
        PACKET_SIZE > 20)
        throw new Exception("winDrw - reached!!");
}
\end{verbatim}

%}}}

\end{document}